\documentclass[twoside,bezier,12pt, reqno]{amsart}
\usepackage{amssymb,latexsym,epsfig,mathrsfs,graphpap,graphics,amsmath,amssymb}
\usepackage{amstext}
\usepackage{amsthm}
\usepackage[utf8]{inputenc}
\usepackage[english]{babel}
\usepackage{helvet}
\usepackage{courier}

\newtheorem{theorem}{Theorem}[section]
\newtheorem{lemma}{Lemma}[section]

\theoremstyle{Definition}
\newtheorem{definition}{Definition}[section]

\theoremstyle{remark}
\newtheorem{remark}[theorem]{Remark}
\numberwithin{equation}{section}

\begin{document}

\begin{flushleft}
 {\bf\Large { Octonion Spectrum of 3D Short-time  LCT Signals  }}

\parindent=0mm \vspace{.2in}

{\bf{M. Younus Bhat$^{1},$ and Aamir H. Dar$^{2}$ }}
\end{flushleft}

{{\it $^{1}$ Department of  Mathematical Sciences,  Islamic University of Science and Technology Awantipora, Pulwama, Jammu and Kashmir 192122, India.E-mail: $\text{gyounusg@gmail.com}$}}

{{\it $^{2}$ Department of  Mathematical Sciences,  Islamic University of Science and Technology Awantipora, Pulwama, Jammu and Kashmir 192122, India.E-mail: $\text{ahdkul740@gmail.com}$}}

\begin{quotation}
\noindent
{\footnotesize {\sc Abstract.} This work is devoted to the development of the octonion linear canonical transform (OLCT) theory proposed by Gao and Li in 2021 that has been designated as an emerging tool in the scenario of signal processing. The purpose of this work is to introduce octonion linear canonical transform of real-valued functions. Further more keeping in mind the varying frequencies, we used the proposed transform to generate a new transform called short-time octonion linear canonical transform (STOLCT).  The results of this article focus on the properties like linearity, reconstruction formula and relation with 3D-short-time linear canonical transform (3D-STLCT).  The crux of this paper lie in establishing well known uncertainty inequalities and convolution theorem for the proposed transform.\.
 \\


}
\end{quotation}
\section{ \bf Introduction}
\noindent

The generalized integral transform called the  linear canonical transform (LCT)  has been designated as an emerging tool in the scenario of signal, image and video  processing recently.  The LCT provides a unified treatment of the generalized Fourier transforms in the sense that it is an embodiment of several well-known integral transforms including the Fourier transform, fractional Fourier transform, Fresnel transform. However, LCT has a drawback. Due to its global kernel it is not suitable for processing the signals with varying frequency content. The short-time linear canonical transform (STLCT) \cite{stlct} with a local window function overcome this drawback. For nonstationary signals STLCT has been used widely and successfully in signal separation  and linear time frequency representation.

The hyper-complex Fourier transform(FT)  is of the great interest in the present era. It treats multi-channel signals as an algebraic whole without losing the spectral relations. Presently, many hyper-complex FTs exists in literature which are defined by different approaches, see \cite{5a,6a}. The developing interest in hyper-complex FTs including applications in  watermarking, color image processing, image filtering,  pattern recognition and edge detection \cite{1b}-\cite{6b}. Among the various hyper-complex FTs, the most basic ones are  the quaternion Fourier transforms(QFTs). QFTs are most widely studied in recent years because of its wide applications in optics and signal processing. Various properties and applications of the QFT were established in \cite{9a}-\cite{12a}. The generalization of quaternion Fourier transform (QFT)is quaternion linear canonical transform (QLCT), which is more effective signal processing tool than QFT due to its extra parameters, see\cite{13a,14a,15a,16a,17a,18a,WLCT}. Later, the quaternion linear canonical transform (QLCT) with four parameters  has been generalized to  short-time quaternion  linear canonical transform (STQLCT) \cite{ucp 2sd qolct}. It is useful in quaternion valued signals and is an alternative to 2D complex STLCT.  Hence has found wide applications in image and signal processing, see \cite{own1,own2,own3,gen}.\\
On the other hand  the Cayley-Dickson algebra of order 8 is known as octonion algebra  which deserve special attention  in the  hyper-complex signal processing. The   octonion Fourier transform (OFT) was proposed by Hahn and Snopek in 2011\cite{hn}. From then OFT is becoming the hot area of research in modern signal processing. Some properties and uncertainty relations and applications associated with OFT have been studied, see\cite{w1,w2,w3,w4}. In 2021 Gao and Li \cite{li} proposed octonion linear canonical transform (OLCT) as a generalization of  OFT by substituting the Fourier kernel with the LCT kernel. They established some vital properties like inversion formula, isometry, Riemann-Lebesgue lemma and proved Heisenberg's and Donoho-Stark's  uncertainty principles.  Furthermore they \cite{stoft} introduced  octonion short-time Fourier transform, where they established classical properties besides establishing Pitt's, Lieb's and uncertainty inequalities.  The generalization of OFT to other transforms  is still in its infancy.

So motivated and inspired  by this, we shall propose the novel  octonion linear canonical transform of real-valued functions. Further more keeping in mind that the OLCT takes signals from time domain to the frequency domain but is unable to perform time-frequency localization simultaneously due to its global kernel. So to overcome this drawback,  we used the proposed transform to generate a new transform called short-time octonion linear canonical transform (STOLCT).  The results of this article focus on the properties like linearity, reconstruction formula and relation with 3D-short-time linear canonical transform (3D-STLCT).  The crux of this paper lie in establishing well known uncertainty inequalities and convolution theorem for the proposed transform..\\
 The highlights of the paper are pointed out below:\\
 \begin{itemize}
 \item  To introduce a novel integral transform coined as the octonion  linear canonical transform (OLCT) for real-valued functions.\\\
 \item To introduce short-time octonion linear canonical transform and decompose  it in to components of different parity.\\\

 \item  To study the properties like linearity and  reconstruction formula.\\\
 \item To study and establish the relationship between short-time octonion linear canonical transform and 3D short-time  linear canonical transform.\\\
  \item To formulate several classes of uncertainty inequalities, such as the Hausdorff-Young inequality,  Lieb's  inequality and  logarithmic uncertainty inequality.\\\
  \item On the basis of classical convolution operation, we establish  convolution theorem for the proposed transform.\\\
      \end{itemize}
      The rest of the paper is organized as follows: In Section 2, some general definitions and basic properties of octonions  are summarized. The definition and the properties of the OLCT  are studied in Section 3. The concept of STOLCT and its associated properties are established in Section 4. In section 5, we develop a series of uncertainty inequalities such as the Hausdorff-Young inequality, Leibs inequality and logarithmic uncertainty inequality associated with the STOLCT. Also   the convolution theorem for the STOLCT is obtained in this section. The potential applications of the STOLCT are presented
in Section 5. In section 6, the   conclusions of the proposed work are drawn.
\section{\bf Preliminaries}
\label{sec 2}
In this section,  we collect some basic facts on the octonion algebra and the offset linear canonical transform(OLCT), which will be needed throughout the paper.

\subsection{\bf Octonion algebra} \  \\
The octonion algebra denoted by $\mathbb O,$ \cite{23} is generated by the eighth-order Cayley-Dickson construction. According to  His
construction, a hypercomplex number  $o\in\mathbb O$ is an ordered pair of quaternions $q_0,q_1\in\mathbb H$
\begin{eqnarray}\label{o1}
\nonumber o&=&(q_0,q_1)\\
\nonumber&=&((z_0,z_1),(z_2,z_3))\\
\nonumber&=&q_0+q_1.\mu_4\\
\nonumber&=&(z_0+z_1.\mu_2)+(z_2+z_3.\mu_2).\mu_4\\
\end{eqnarray}
which has equivalent form
\begin{equation}\label{o2}
o=s_o+\sum_{i=1}^7s_i\mu_i=s_0+s_1\mu_1+s_2\mu_2+s_3\mu_3+s_4\mu_4+s_5\mu_5+s_6\mu_6+s_7\mu_7
\end{equation}
that is $o$ is a hypercomplex number defined by eight real numbers $s_i,i=0,1,\dots,7$ and seven imaginary units $\mu_i$ where $i=1,2,\dots,7.$ The octonion algebra is non-commutative and non-associative algebra. The multiplication
of imaginary units in the Cayley-Dickson algebra of octonions are presented in Table I .\cite{w2}\\
\begin{center} Table I \end{center}
\begin{center}{\small Multiplication Rules in Octonion Algebra.} \end{center}
\begin{equation*}\label{table}
\begin{array}{|c|c|c|c|c|c|c|c|c|}
\hline
  \cdot & 1 & \mu_1 & \mu_2 & \mu_3 & \mu_4 & \mu_5 & \mu_6 & \mu_7 \\
\hline
  1 & 1 & \mu_1 & \mu_2 & \mu_3 & \mu_4 & \mu_5 & \mu_6 & \mu_7 \\
\hline
  \mu_1 & \mu_1 & -1 & \mu_3 & -\mu_2 & \mu_5 & -\mu_4& -\mu_7& \mu_6 \\
\hline
  \mu_2 & \mu_2 & -\mu_3 & -1 & \mu_1 & \mu_6 & \mu_7 & -\mu_4 & -\mu_5 \\
\hline
  \mu_3 & \mu_3 & \mu_2 & -\mu_1 & -1 & \mu_7 & -\mu_6 & \mu_5 & -\mu_4 \\
\hline
  \mu_4 & \mu_4 & -\mu_5 & -\mu_6 & -\mu_7 & -1 & \mu_1 & \mu_2 & \mu_3 \\
  \hline
  \mu_5 & \mu_5 & \mu_4 & -\mu_7 & \mu_6 & -\mu_1 & -1 & -\mu_3 & \mu_2\\
\hline
  \mu_6 & \mu_6 & \mu_7 & \mu_4 & -\mu_5 & -\mu_2 & \mu_3 & -1 & -\mu_1 \\
  \hline
  \mu_7 & \mu_7 & -\mu_6 & \mu_5& \mu_4 & -\mu_3 & -\mu_2 & \mu_1 & -1 \\
  \hline

\end{array}\end{equation*}
The conjugate of an octonion is defined as
\begin{equation}\label{o3}
\overline{o}=s_0-s_1\mu_1-s_2\mu_2-s_3\mu_3-s_4\mu_4-s_5\mu_5-s_6\mu_6-s_7\mu_7
 \end{equation}
 Therefore norm is defined by $|o|=\sqrt{o\overline{o}}$ and $|o|^2=\sum_{i=o}^7s_i.$ Also $|o_1o_2|=|o_1||o_2|,\forall o_1,o_2\in \mathbb O.$

From (\ref{o1}) it is evident that every $o\in \mathbb O$ can be represented in quaternion form as
\begin{equation}\label{o4}
o=a+b\mu_4
\end{equation}
where $a=s_0+s_1\mu_1+s_2\mu_2+s_3\mu_3$ and $b=s_4+s_5\mu_1+s_6\mu_2+s_7\mu_3$ are both quaternions. By direct verification we have following lemma.
\begin{lemma}\cite{w2}\label{lem1} Let $a,b\in \mathbb H,$ then\\
(1)\quad $\mu_4a=\overline a\mu_4;$ \qquad\quad(2)\quad$\mu_4(a\mu_4)=-\overline a;$\qquad\quad$(3)\quad(a\mu_4)\mu_4=-a;$\\
(4)\quad$a(b\mu_4)=(ba)\mu_4$;\quad\quad(5)\quad$(a\mu_4)b=(a\overline b)\mu_4;$\quad\quad(6)\quad$(a\mu_4)(b\mu_4)=-\overline b a.$
\end{lemma}
It is clear from above Lemma that, for an octonion $a+b\mu_4, a,b \in \mathbb H,$ we have \\
\begin{equation}\label{o5} \overline{a+b\mu_4}=\overline{a}-b\mu_4\end{equation}
and \begin{equation}\label{o6} |a+b\mu_4|^2=|a|^2+|b|^2.\end{equation}
\begin{lemma}\label{lem1}
Let $\tilde o, \hat o\in\mathbb O$. Then $e^{\tilde o}.e^{\hat o}=e^{\tilde o+\hat o}$ iff $\tilde o.\hat o=\hat o.\tilde o$.
\end{lemma}
An octonion-valued function $f:\mathbb R^3\longrightarrow\mathbb O$ has following explicit form
\begin{eqnarray}\label{ofun}
\nonumber f(x)&=&f_0+f_1(x)\mu_1+f_2(x)\mu_2+f_3(x)\mu_3+f_4(x)\mu_4+f_5(x)\mu_5+f_6(x)\mu_6+f_7(x)\mu_7\\
\nonumber&=&f_0+f_1\mu_1+(f_2+f_3\mu_1)\mu_2+[f_4+f_5\mu_1+(f_6+f_7\mu_1)\mu_2]\mu_4\\
&=&g(x)+h(x)\mu_4
\end{eqnarray}
where each $f_i(x)$ is a real valued functions, $g,h \in \mathbb H$ are as in(\ref{o1}) and $x=(x_1,x_2,x_3)\in \mathbb R^3.$

For  each real-valued  function $f(x)$ over $\mathbb R^k$ and $1\le p<\infty,$ the $L^p-$norm of $f$ is defined by
\begin{equation}\label{lpnorm}
\|f\|_{L^p(\mathbb R^k)}=\left(\int_{\mathbb R^k}|f(x)|^pdx\right)^{\frac{1}{p}},
\end{equation}
where $x=(x_1,x_2,...,x_k)\in\mathbb R^d$
And for $p=\infty$, then the $L^\infty$-norm is defined by
\begin{equation}\label{linfty}
\|f\|_\infty=esssup_{x\in\mathbb R^k}|f(x)|.
\end{equation}
For any functions  $f(x),g(x)$ over $\mathbb R^k,$ the innear product is given by
\begin{equation}\label{ips}
\langle f,g\rangle_{L^2(\mathbb R^k)}=\int_{\mathbb R^k}f(x)\overline{g(x)}dx.
\end{equation}
Let $f,g\in L^2(\mathbb R^k),$ the classic convolution operation is defined as
\begin{equation}\label{con op}
(f\ast g)(x)=\int_{\mathbb R^k}f(y)g(x-y)dy.
\end{equation}
\subsection{\bf Octonion Linear Canonical Transform of Octonion-valued Functions} \  \\

 In 2021 Gao,W.B and Li,B.Z \cite{li} introduced linear canonical transform in octonion setting they called it the octonion linear canonical transform (OLCT) and defined it as follows:\\
 For $f\in L^1(\mathbb R^3,\mathbb O),$ then  the one dimensional  OLCT with respect to the uni-modular  matrix $A=(a,b,c,d)$ is given by
\begin{equation}\label{onedOLCT}
\mathcal L^{A}_{\mu_4}\{f\} (w)= \int_{\mathbb R} f(x)K_{A}^{\mu_4}(x,w)dx,
\end{equation}
where
 \begin{equation*}\label{eqnolctk}
K^{\mu_4}_{A}(x,w)=\dfrac{1}{\sqrt{2\pi |b|}} e^{\frac{{\mu_4}}{2b}\big[ax^2-2xw-+dw^2-\frac{\pi}{2}\big]},\quad b\neq0\,
\end{equation*}
 with the inversion formula
 \begin{eqnarray}
  f(x)=\int_{\mathbb R}\mathcal L^{A}_{\mu_4}\{f\}(w)K_{A}^{-\mu_4}(x,w)dx,
 \end{eqnarray}
 where $K_{A}^{-\mu_4}(x,w)=K_{A^{-1}}^{\mu_4}(w,x)$ and $A^{-1}=(d,-b,-c,a).$\\

  And for octonion valued function $f\in L^1(\mathbb R^3,\mathbb O)\cap L^2(\mathbb R^3,\mathbb O)$, the three dimensional  OLCT with respect to the   matrix parameter $A_k=(a_k,b_k,c_k,d_k), $ satisfying $det(A_k)=1,\quad k=1,2,3$ is defined as
 \begin{equation}\label{threedOLCT}
\mathcal L^{A_1,A_2,A_3}_{\mu_1,\mu_2,\mu_4}\{f\}(w)=\int_{\mathbb R^3}f(x) K^{\mu_1}_{A_1}(x_1,w_1)K^{\mu_2}_{A_1}(x_2,w_2)K^{\mu_4}_{A_3}(x_3,w_3)dx
\end{equation}
where $x=(x_{1},x_{2},x_3),\, w=(w_{1},w_{2},w_3),$ and  multiplication in above integral is done from left to right and
\begin{equation*}
K_{A_1}^{\mu_1}(x_1,w_1)=\dfrac{1}{\sqrt{2\pi |b_1|}} e^{\frac{{\mu_1}}{2b_1}\big[a_1x^2_1-2x_1w_1+d_1w^2_1-\frac{\pi}{2}\big]},\quad b_1\neq0\,
\end{equation*}
 \begin{equation*}
K_{A_2}^{\mu_2}(x_2,w_2)==\dfrac{1}{\sqrt{2\pi |b_2|}} e^{\frac{{\mu_2}}{2b_2}\big[a_2x^2_2-2x_2w_2+d_2w^2_2-\frac{\pi}{2}\big]},\quad b_2\neq0\,
\end{equation*}
and
 \begin{equation*}
K_{A_3}^{\mu_4}(x_3,w_3)=\dfrac{1}{\sqrt{2\pi |b_3|}} e^{\frac{{\mu_4}}{2b_3}\big[a_3x^2_3-2x_3w_3+d_3w^2_3-\frac{\pi}{2}\big]},\quad b_3\neq0 .
\end{equation*}
with the inversion formula
\begin{equation}\label{eqninv}
f(x)=\int_{\mathbb R^3}\mathcal L^{A_1,A_2,A_3}_{\mu_1,\mu_2,\mu_4}\{f\}(w)K_{A^{-1}_3}^{\mu_4}(w_3,x_3)K_{A^{-1}_2}^{\mu_2}(w_2,x_2)K_{A^{-1}_1}^{\mu_1}(w_1,x_1)dx,
\end{equation}
where $A^{-1}_k=(d_k,-b_k,-c_k,a_k)\in\mathbb R^{2\times 2},$ for $ k=1,2,3.$

\section{\bf Octonion Linear Canonical Transform of Real-valued Functions} \  \\
According to the  octonion Fourier transform(OFT)\cite{w1} of a real-valued functions of
three variables and octonion linear canonical transform for octonion valued functions\cite{li} we can obtain the definition of the octonion linear canonical transform(OLCT) of real valued function $f(x)$ in three variables  as follows:
\begin{definition} The 3D-OLCT of real-valued function $f:\mathbb R^3\rightarrow\mathbb R$ can be defined as
\begin{equation}\label{realOLCT}
\mathcal L^{A_1,A_2,A_3}_{\mu_1,\mu_2,\mu_4}\{f\}(w)=\int_{\mathbb R^3}f(x) K^{\mu_1}_{A_1}(x_1,w_1)K^{\mu_2}_{A_1}(x_2,w_2)K^{\mu_4}_{A_3}(x_3,w_3)dx.
\end{equation}
\end{definition}
where $x=(x_{1},x_{2},x_3),\, w=(w_{1},w_{2},w_3),$ and  kernel signals

\begin{equation}\label{k1}
K_{A_1}^{\mu_1}(x_1,w_1)=\dfrac{1}{\sqrt{2\pi |b_1|}} e^{\frac{{\mu_1}}{2b_1}\big[a_1x^2_1-2x_1w_1+d_1w^2_1-\frac{\pi}{2}\big]},\quad b_1\neq0\,
\end{equation}
 \begin{equation}\label{k2}
K_{A_2}^{\mu_2}(x_2,w_2)==\dfrac{1}{\sqrt{2\pi |b_2|}} e^{\frac{{\mu_2}}{2b_2}\big[a_2x^2_2-2x_2w_2+d_2w^2_2-\frac{\pi}{2}\big]},\quad b_2\neq0\,
\end{equation}
and
 \begin{equation}\label{k3}
K_{A_3}^{\mu_4}(x_3,w_3)=\dfrac{1}{\sqrt{2\pi |b_3|}} e^{\frac{{\mu_4}}{2b_3}\big[a_3x^2_3-2x_3w_3+d_3w^2_3-\frac{\pi}{2}\big]},\quad b_3\neq0 .
\end{equation}
Since the octonion algebra is non-associative it should be noted that the
multiplication in the above integrals is done from left to right. Also we assume that the above signal $f$ is continuous and both signal and its OLCT are integrable(in Lebesgue sense) in this paper.
\begin{theorem}[Inversion]\label{th inv real olct} Let $f:\mathbb R^3\rightarrow\mathbb R$ be a continuous and square-integrable
function (in Lebesgue sense). Then the OLCT of $f$ is an invertible, and its inverse if given by
\begin{eqnarray}\label{eqn inv real}
\nonumber f(x)&=&\{\mathcal L^{A_1,A_2,A_3}_{\mu_1,\mu_2,\mu_4}\}^{-1}\left( L^{A_1,A_2,A_3}_{\mu_1,\mu_2,\mu_4}\{f\}(x)\right)\\
&=&\int_{\mathbb R^3}\mathcal L^{A_1,A_2,A_3}_{\mu_1,\mu_2,\mu_4}\{f\}(w)K_{A^{-1}_3}^{\mu_4}(w_3,x_3)K_{A^{-1}_2}^{\mu_2}(w_2,x_2)K_{A^{-1}_1}^{\mu_1}(w_1,x_1)dx,
\end{eqnarray}
where $A^{-1}_k=(d_k,-b_k,-c_k,a_k)\in\mathbb R^{2\times 2},$ for $ k=1,2,3.$
\end{theorem}
\begin{proof}
Consider the octonion-valued function $f:\mathbb R^3\rightarrow\mathbb O,$ i.e.  $$f(x)=f_0+f_1(x)\mu_1+f_2(x)\mu_2+f_3(x)\mu_3+f_4(x)\mu_4+f_5(x)\mu_5+f_6(x)\mu_6+f_7(x)\mu_7,$$
where $f_i:\mathbb R^3\rightarrow\mathbb R$, $i=0,1,2,3,4,5,6,7.$
we have by (\ref{eqninv})
\begin{equation*}
f(x)=\int_{\mathbb R^3}\mathcal L^{A_1,A_2,A_3}_{\mu_1,\mu_2,\mu_4}\{f\}(w)K_{A^{-1}_3}^{\mu_4}(w_3,x_3)K_{A^{-1}_2}^{\mu_2}(w_2,x_2)K_{A^{-1}_1}^{\mu_1}(w_1,x_1)dx,
\end{equation*}
Thus for the special case $f:\mathbb R^3\rightarrow\mathbb R$ result follows.

Note the result also follows by using the procedure of Theorem 3.1\cite{w1}.
\end{proof}
Further, we can expand the kernel of the  OLCT in the form
\begin{eqnarray}\label{eulerproduct}
 \nonumber K_{A_1}^{\mu_1}(x_1,w_1) K_{A_2}^{\mu_2}(x_2,w_2) K_{A_3}^{\mu_4}(x_3,w_3)&=&\frac{1}{2\pi\sqrt{2\pi|b_1b_2b_3|}}e^{\mu_1\xi_1}e^{\mu_2\xi_2}e^{\mu_4\xi_3}\\
\nonumber&=&\frac{1}{2\pi\sqrt{2\pi|b_1b_2b_3|}}(c_1+\mu_1s_1)(c_2+\mu_2s_2)(c_3+\mu_4s_3)\\
\nonumber&=&\frac{1}{2\pi\sqrt{2\pi|b_1b_2b_3|}}(c_1c_2c_3+s_1c_2c_3\mu_1+c_1s_2c_3\mu_2\\
\nonumber&\quad+& s_1s_2c_3\mu_3+c_1c_2s_3\mu_4+s_1c_2s_3\mu_5+c_1s_2s_3\mu_6+s_1s_2s_3\mu_7),\\
\end{eqnarray}
where $\xi_k={\frac{{1}}{2b_k}\big[a_kx^2_k-2x_kw_k+d_kw^2_k-\frac{\pi}{2}\big]},\quad c_k=\cos\xi_k$ and $s_k=\sin\xi_k,\quad k=1,2,3.$\\

 Now using (\ref{eulerproduct}) OLCT of a real-valued function $\mathcal L^{A_1,A_2,A_3}_{\mu_1,\mu_2,\mu_4}\{f\}(w)$ of three variables can be expressed as octonion sum of
components of different parity  (\cite{1b,hn}):
\begin{eqnarray}\label{olct parity}
\nonumber\mathcal L^{A_1,A_2,A_3}_{\mu_1,\mu_2,\mu_4}\{f\}(w)&=&L_{eee}+L_{oee}\mu_1+L_{eoe}\mu_2+L_{ooe}\mu_3+L_{eeo}\mu_4\\
\nonumber&&\qquad\qquad\qquad+L_{oeo}\mu_5+L_{eoo}\mu_6+L_{ooo}\mu_7\\\
\end{eqnarray}
where
\begin{equation*}
L_{eee}(w)=\frac{1}{2\pi\sqrt{2\pi|b_1b_2b_3|}}\int_{\mathbb R^3}f_{eee}(x)\cos\xi_1\cos\xi_2\cos\xi_3dx,
\end{equation*}
\begin{equation*}
L_{oee}(w)=\frac{1}{2\pi\sqrt{2\pi|b_1b_2b_3|}}\int_{\mathbb R^3}f_{oee}(x)\sin\xi_1\cos\xi_2\cos\xi_3dx,
\end{equation*}

\begin{equation*}
L_{eoe}(w)=\frac{1}{2\pi\sqrt{2\pi|b_1b_2b_3|}}\int_{\mathbb R^3}f_{eoe}(x,u)\cos\xi_1\sin\xi_2\cos\xi_3dx,
\end{equation*}

\begin{equation*}
L_{ooe}(w)=\frac{1}{2\pi\sqrt{2\pi|b_1b_2b_3|}}\int_{\mathbb R^3}f_{ooe}(x)\sin\xi_1\sin\xi_2\cos\xi_3dx,
\end{equation*}

\begin{equation*}
L_{eeo}(w)=\frac{1}{2\pi\sqrt{2\pi|b_1b_2b_3|}}\int_{\mathbb R^3}f_{eeo}(x)\cos\xi_1\cos\xi_2\sin\xi_3dx,
\end{equation*}

\begin{equation*}
L_{oeo}(w)=\frac{1}{2\pi\sqrt{2\pi|b_1b_2b_3|}}\int_{\mathbb R^3}f_{oeo}(x)\sin\xi_1\cos\xi_2\sin\xi_3dx,
\end{equation*}

\begin{equation*}
L_{eoo}(w)=\frac{1}{2\pi\sqrt{2\pi|b_1b_2b_3|}}\int_{\mathbb R^3}f_{eoo}(x)\cos\xi_1\cos\xi_2\sin\xi_3dx,
\end{equation*}

\begin{equation*}
L_{ooo}(w)=\frac{1}{2\pi\sqrt{2\pi|b_1b_2b_3|}}\int_{\mathbb R^3}f_{ooo}(x)\sin\xi_1\sin\xi_2\sin\xi_3dx.
\end{equation*}

Where $f_{lmn}(x,u),l,m,n\in\{e,o\}$ are eight terms of different parity with relation to $x_1,
x_2$ and $x_3.$ In the above notation, we use subscripts $e$ and $o$ to indicate that a
function is either even (e) or odd (o) with respect to an appropriate
variable, i.e. $f_{eeo}(x)$ is even with respect to $x_1$ and $x_2$ and odd with
respect to $x_3$.

Before moving forward we introduce the   3D-LCT.

 \begin{definition}\cite{li}\label{def 3dlct}
 The 3D-LCT is defined by
\begin{equation}\label{eqn 3dlct}
\mathcal L_{A_1,A_2,A_3}\{f\}(w)=\dfrac{1}{2\pi\sqrt{2\pi|b_1b_2b_3|}}\int_{\mathbb R^3}f(x)e^{\mu_1\xi_1}e^{\mu_1\xi_2}e^{\mu_1\xi_3}dx.
\end{equation}
\end{definition}

\begin{lemma}\label{3dlct}The relation between OLCT and 3-D LCT is that
\begin{eqnarray}
\nonumber \mathcal L^{A_1,A_2,A_3}_{\mu_1,\mu_2,\mu_4}\{f\}(w)=&\dfrac{1}{4}\left\{(\mathcal L_{A_1,A_2,A_3}\{f\}(w)+\mathcal L_{A_1,A_2,A'_3}\{f\}(w))(1-\mu_3)\right.\\
\nonumber\qquad &\left.+(\mathcal L_{A_1,A'_2,A_3}\{f\}(w)+\mathcal L_{A_1,A'_2,A'_3}\{f\}(w))(1+\mu_3)\right\}\\
\nonumber\qquad\quad& +\dfrac{1}{4}\left\{(\mathcal L_{A_1,A_2,A_3}\{f\}(w)-\mathcal L_{A_1,A_2,A'_3}\{f\}(w))(1-\mu_3)\right.\\
\label{3dlcteqn}\qquad &\left.+(\mathcal L_{A_1,A'_2,A_3}\{f\}(w)-\mathcal L_{A_1,A'_2,A'_3}\{f\}(w))(1+\mu_3)\right\}.\mu_5
\end{eqnarray}
where $A'_k=(a_k,-b_k,-c_k,d_k),\quad k=1,2.$
\end{lemma}
\begin{proof}The proof is similar to the Theorem 6\cite{li}.
\end{proof}
\subsection{\bf Quaternion Short-Time Linear Canonical Transform} \  \\
The Quaternion Short-Time Linear Canonical Transform(QSTLCT) was introduced by Zhu and Zheng, which is a generalization of the Short-Time Linear Canonical Transform(STLCT) in the quaternion algebra setting\cite{ucp 2sd qolct}. Let $\mu_1,\mu_2$ and $\mu_3$ (or equivalently i,j,k) denote the three imaginary units in quaternion algebra.

For $A_i=(a_i,b_i,c_i,d_i)\in \mathbb R^{2\times 2}$ be a matrix parameter
satisfying det$(A_i)=1$, for $i=1,2.$ Let  $\phi\in L^2(\mathbb R^2,\mathbb H)$ be a non-zero quaternion window function. Then (QSTLCT) of a signal $f\in L^2(\mathbb R^2,\mathbb H)$ can be defined as
\begin{equation}
\mathcal G^{A_1,A_2}_{\phi}f(w,u)=\int_{\mathbb R^2}f(x)\overline{\phi(x-u)}K^{\mu_1}_{A_1}(x_1,w_1)K^{\mu_2}_{A_2}(x_2,w_2)dx,
\end{equation}
where $x=(x_1,x_2)\in\mathbb R^2,$ $w=(w_1,w_2)\in\mathbb R^2$,  $u=(u_1,u_2)\in\mathbb R^2$ and $K^{\mu_1}_{A_1}(x_1,w_1),$ $K^{\mu_2}_{A_2}(x_2,w_2)$ are given by equations (\ref{k1}) and (\ref{k2}) respectively.

\section{\bf Short-Time Octonion Linear Canonical Transform}

 In this section, we define the novel short-time octonion linear cananical transform (STOLCT) of real valued function of three variables and discuss several basic properties of the STOLCT. These
properties play important roles in signal representation.
\begin{definition}\label{def stolct}
Let $A_i=(a_i,b_i,c_i,d_i)\in \mathbb R^{2\times 2}$ be a matrix parameter
satisfying det$(A_i)=1$, for $i=1,2,3.$   Let  $\phi\in L^2(\mathbb R^3)$ be a non-zero real-valued window function. Then STOLCT of a real-valued  signal $f\in L^2(\mathbb R^3)$ can be defined as
\end{definition}
\begin{equation}\label{eqn stolct}
\mathcal G^{A_1,A_2,A_3}_{\phi}\{f\}(w,u)=\int_{\mathbb R^3}f(x)\overline{\phi(x-u)}K^{\mu_1}_{A_1}(x_1,w_1)K^{\mu_2}_{A_2}(x_2,w_2)K^{\mu_4}_{A_3}(x_3,w_3)dx,
\end{equation}
where $x=(x_{1},x_{2},x_3)\in\mathbb R^3,\, w=(w_{1},w_{2},w_3)\in\mathbb R^3,$ and  $u=(u_{1},u_{2},u_3)\in\mathbb R^3$ and
\begin{equation}\label{k1}
K_{A_1}^{\mu_1}(x_1,w_1)=\dfrac{1}{\sqrt{2\pi |b_1|}} e^{\frac{{\mu_1}}{2b_1}\big[a_1x^2_1-2x_1w_1+d_1w^2_1-\frac{\pi}{2}\big]},\quad b_1\neq0\,
\end{equation}
 \begin{equation}\label{k2}
K_{A_2}^{\mu_2}(x_2,w_2)==\dfrac{1}{\sqrt{2\pi |b_2|}} e^{\frac{{\mu_2}}{2b_2}\big[a_2x^2_2-2x_2w_2+d_2w^2_2-\frac{\pi}{2}\big]},\quad b_2\neq0\,
\end{equation}
and
 \begin{equation}\label{k3}
K_{A_3}^{\mu_4}(x_3,w_3)=\dfrac{1}{\sqrt{2\pi |b_3|}} e^{\frac{{\mu_4}}{2b_3}\big[a_3x^2_3-2x_3w_3+d_3w^2_3-\frac{\pi}{2}\big]},\quad b_3\neq0 .
\end{equation}
are kernel signals. $w$ and $u$ represent frequency and time respectively. Since $\phi$ is real-valued so complex conjugate  $\overline\phi=\phi.$
On the basis of classical convolution STOLCT defined in (\ref{eqn stolct}) can be rewritten as
\begin{equation*}
\mathcal G^{A_1,A_2,A_3}_{\phi}\{f\}(w,u)=\left(f(u)K^{\mu_1}_{A_1}(x_1,w_1)K^{\mu_2}_{A_2}(x_2,w_2)K^{\mu_4}_{A_3}(x_3,w_3)\right)\ast\left(\overline{\phi(-u)}\right).
\end{equation*}
\begin{remark}\label{rem1}It should be noted that the
multiplication in the above integrals is performed from
left to right, as the octonion
algebra is non-associative.
\end{remark}
\begin{remark}\label{rem2}
With the help of quaternion and octonion algebra the formula (\ref{eqn stolct}) can be re written as
\begin{equation*}
\mathcal G^{A_1,A_2,A_3}_{\phi}\{f\}(w,u)=\langle f,\Phi_{x,w,u}\rangle,
\end{equation*}
where $\Phi_{x,w,u}={\phi(x-u)}K^{\mu_4}_{A_3}(x_3,w_3)\left(K^{\mu_2}_{A_2}(x_2,w_2)K^{\mu_1}_{A_1}(x_1,w_1)\right)$ is
the kernel of the STOLCT.
\end{remark}
Moreover Definition \ref{def stolct} can be expressed as
\begin{eqnarray}\label{eqn stolct new}
\nonumber\mathcal G^{A_1,A_2,A_3}_{\phi}\{f\}(w,u)&=\mathcal L^{A_1,A_2,A_3}_{\mu_1,\mu_2,\mu_4}\{f(x)\overline{\phi(x-u)}\}(w)\\
\nonumber&=\mathcal L^{A_1,A_2,A_3}_{\mu_1,\mu_2,\mu_4}\{h\}(w),\\
\end{eqnarray}
where $h(x,u)=f(x)\overline{\phi(x-u)}.$

It is clear from (\ref{eqn stolct new}) that STOLCT of a signal is a two step process. In first step signal is multiplied by a window function and then in second step we obtain OLCT of multiplied signal. Thus all the results for OLCT can be extended to the novel STOLCT, and vice versa.

\subsection{\bf Decomposition of STOLCT in to components of different parity}

The STOLCT of a real-valued function $\mathcal G^{A_1,A_2,A_3}_{\phi}\{f\}(w,u)$ of three variables can be expressed (using (\ref{eulerproduct}) and (\ref{olct parity})) as octonion sum of
components of different parity :
\begin{eqnarray}\label{parity}
\nonumber\mathcal G^{A_1,A_2,A_3}_{\phi}\{f\}&=&G^\phi_{eee}+G^\phi_{oee}\mu_1+G^\phi_{eoe}\mu_2+G^\phi_{ooe}\mu_3+G^\phi_{eeo}\mu_4\\
\nonumber&&\qquad\qquad\qquad+G^\phi_{oeo}\mu_5+G^\phi_{eoo}\mu_6+G^\phi_{ooo}\mu_7\\\
\end{eqnarray}
where
\begin{equation}\label{e1}
G^\phi_{eee}(w,u)=\frac{1}{2\pi\sqrt{2\pi|b_1b_2b_3|}}\int_{\mathbb R^3}h_{eee}(x,u)\cos\xi_1\cos\xi_2\cos\xi_3dx,
\end{equation}
\begin{equation}\label{e2}
G^\phi_{oee}(w,u)=\frac{1}{2\pi\sqrt{2\pi|b_1b_2b_3|}}\int_{\mathbb R^3}h_{oee}(x,u)\sin\xi_1\cos\xi_2\cos\xi_3dx,
\end{equation}

\begin{equation}\label{e3}
G^\phi_{eoe}(w,u)=\frac{1}{2\pi\sqrt{2\pi|b_1b_2b_3|}}\int_{\mathbb R^3}h_{eoe}(x,u)\cos\xi_1\sin\xi_2\cos\xi_3dx,
\end{equation}

\begin{equation}\label{e4}
G^\phi_{ooe}(w,u)=\frac{1}{2\pi\sqrt{2\pi|b_1b_2b_3|}}\int_{\mathbb R^3}h_{ooe}(x,u)\sin\xi_1\sin\xi_2\cos\xi_3dx,
\end{equation}

\begin{equation}\label{e5}
G^\phi_{eeo}(w,u)=\frac{1}{2\pi\sqrt{2\pi|b_1b_2b_3|}}\int_{\mathbb R^3}h_{eeo}(x,u)\cos\xi_1\cos\xi_2\sin\xi_3dx,
\end{equation}

\begin{equation}\label{e6}
G^\phi_{oeo}(w,u)=\frac{1}{2\pi\sqrt{2\pi|b_1b_2b_3|}}\int_{\mathbb R^3}h_{oeo}(x,u)\sin\xi_1\cos\xi_2\sin\xi_3dx,
\end{equation}

\begin{equation}\label{e7}
G^\phi_{eoo}(w,u)=\frac{1}{2\pi\sqrt{2\pi|b_1b_2b_3|}}\int_{\mathbb R^3}h_{eoo}(x,u)\cos\xi_1\cos\xi_2\sin\xi_3dx,
\end{equation}

\begin{equation}\label{e8}
G^\phi_{ooo}(w,u)=\frac{1}{2\pi\sqrt{2\pi|b_1b_2b_3|}}\int_{\mathbb R^3}h_{ooo}(x,u)\sin\xi_1\sin\xi_2\sin\xi_3dx.
\end{equation}

 Where $h(x,u)=f(x)\phi(x-u)$ is a real valued function and it can be expressed as  sum  eight terms:
 \begin{eqnarray}\label{h even odd}
 \nonumber h(x,u)&=&h_{eee}(x,u)+h_{eeo}(x,u)+h_{eoe}(x,u)+h_{eoo}(x,u)\\\
 \nonumber&&+h_{oee}(x,u)+h_{oeo}(x,u)+h_{ooe}(x,u)+h_{ooo}(x,u),\\\
 \end{eqnarray}
where $h_{lmn}(x,u),l,m,n\in\{e,o\}$ are eight terms of different parity with relation to $x_1,
x_2$ and $x_3.$ Again using subscripts $e$ and $o$ to indicate that a
function is either even (e) or odd (o) with respect to an appropriate
variable, i.e. $h_{eeo}(x)$ is even with respect to $x_1$ and $x_2$ and odd with
respect to $x_3$.

Now, we discuss several basic properties of the ST-QOLCT given by (3.1). These properties play important roles in signal representation.\\

\begin{theorem}[Linearity]\label{th1}Let $f,g\in L^2(\mathbb R^3)$ be two real-valued signals and $\phi$ be a real-valued non-zero window function in $L^2(\mathbb R^3).$ Then for $\alpha,\beta\in\mathbb R,$ we have
\begin{equation}
\mathcal G_{\phi}^{A_1,A_2,A_3}\{\alpha f+\beta g\}=\alpha\mathcal G_{\phi}^{A_1,A_2,A_3}\{f\}+\beta \mathcal G_{\phi}^{A_1,A_2,A_3}\{g\}.
\end{equation}
\end{theorem}
\begin{proof}This follows directly from the linearity of the product and the integration
involved in Definition \ref{def stolct}.
\end{proof}

 \begin{lemma}\label{lem pq2} Let $\phi\in L^p(\mathbb{R}^3),f\in L^1\left( \mathbb{R}^3\right)$. Then we have
\begin{equation}\label{eqn pq2}
\|\mathcal G_{\phi}^{A_1,A_2,A_3}\{f\}(w,u)\|_{L^{p}(\mathbb R^3)}\leq \frac{1}{2\pi\sqrt{2\pi|b_{1}b_{2}b_3|}} \|f\|_{L^{1}(\mathbb R^3)}. \|\phi\|_{L^{q}(\mathbb R^3)}.
\end{equation}
\end{lemma}
\begin{proof}
By the virtue of  Minkowski's inequality , we have
\begin{eqnarray*}
\|\mathcal G_{\phi}^{A_1,A_2,A_3}\{f\}(w,u)\|_{L^{p}(\mathbb R^3)}&\leq&\int_{\mathbb R^3}\left(\int_{\mathbb R^3}\left|f(x)\overline{\phi(x-u)}K^{\mu_1}_{A_1}(x_1,w_1)K^{\mu_2}_{A_2}(x_2,w_2)\right.\right.\\
&&\times\left.\left.K^{\mu_4}_{A_3}(x_3,w_3)\right|^pdu\right)^{1/p}dx\\
&=&\dfrac{1}{2\pi\sqrt{2\pi|b_{1}b_{2}b_3|}}\int_{\mathbb R^3}\left(\int_{\mathbb R^3}\left|f(x)\overline{\phi(x-u)}\right|^pdu\right)^{1/p}dx.
\end{eqnarray*}
On setting $x-u=z$ in the above inequality, we obtain
\begin{eqnarray*}
\|\mathcal G_{\phi}^{A_1,A_2,A_3}\{f\}(w,u)\|_{L^{p}(\mathbb R^3)}&\leq&\dfrac{1}{2\pi\sqrt{2\pi|b_{1}b_{2}b_3|}}\int_{\mathbb R^3}\left(\left|f(x)\overline{\phi(z)}\right|^p dz\right)^{1/p}dx\\
&=&\dfrac{1}{2\pi\sqrt{2\pi|b_{1}b_{2}b_3|}}\left(\int_{\mathbb R^3}|\overline{\phi(z)}|^pdz\right)^{1/p}\int_{\mathbb R^3}|f(x)|dx\\
&=&\frac{1}{2\pi\sqrt{2\pi|b_{1}b_{2}b_3|}} \|f\|_{L^{1}(\mathbb R^3)}. \|\phi\|_{L^{q}(\mathbb R^3)}.
\end{eqnarray*}
Which completes the proof.
\end{proof}

The next theorem guarantees the reconstruction of the input signal from the corresponding STOLCT.

\begin{theorem}[Reconstruction formula]\label{th2} Let $\phi\in L^2(\mathbb R^3)$ be a non-zero window function, then every real-valued signal  $f\in L^2(\mathbb R^3)$ can be fully reconstructed by the formula
\begin{eqnarray}\label{eqn recons}
\nonumber f(x)&=&\frac{1}{\|\phi\|^2_{L^2(\mathbb R^3)}}\int_{\mathbb R^3}\int_{\mathbb R^3}\mathcal G_{\phi}^{A_1,A_2,A_3}\{f\}(w,u)K^{\mu_4}_{A^{-1}_3}(w_3,x_3)K^{\mu_2}_{A^{-1}_2}(w_2,x_2)\\
\label{eqn recons}&&\qquad\qquad\qquad\qquad\times K^{\mu_1}_{A^{-1}_1}(w_1,x_1)\phi(x-u)dwdu.
\end{eqnarray}
\end{theorem}
\begin{proof}
From (\ref{eqn stolct new}),we have
\begin{equation}\label{a1}
\mathcal G^{A_1,A_2,A_3}_{\phi}\{f\}(w,u)=\mathcal L^{A_1,A_2,A_3}_{\mu_1,\mu_2,\mu_4}\{f(x)\overline{\phi(x-u)}\}(w).
\end{equation}
Applying Theorem \ref{th inv real olct} to (\ref{a1}), we get
\begin{eqnarray}
\nonumber f(x)\overline{\phi(x-u)}&=&\{\mathcal L^{A_1,A_2,A_3}_{\mu_1,\mu_2,\mu_4}\}^{-1}\left(\mathcal G^{A_1,A_2,A_3}_{\phi}\{f\}\right)(x)\\
\label{a2}&=&\int_{\mathbb R^3}\mathcal G^{A_1,A_2,A_3}_{\phi}\{f\}(w,u)K_{A^{-1}_3}^{\mu_4}(w_3,x_3)K_{A^{-1}_2}^{\mu_2}(w_2,x_2)K_{A^{-1}_1}^{\mu_1}(w_1,x_1)dw.
\end{eqnarray}
On multiplying both sides of (\ref{a2}) from right by $\phi(x-u)$, we get
\begin{equation}\label{a3}
 f(x)\overline{\phi(x-u)}\phi(x-u)=\int_{\mathbb R^3}\mathcal G^{A_1,A_2,A_3}_{\phi}\{f\}(w,u)K_{A^{-1}_3}^{\mu_4}(w_3,x_3)K_{A^{-1}_2}^{\mu_2}(w_2,x_2)K_{A^{-1}_1}^{\mu_1}(w_1,x_1)\phi(x-u)dw.
\end{equation}
Now integrating both sides of (\ref{a3}) with respect $du$ and using Fubini's theorem, we obtain
\begin{eqnarray}
\nonumber f(x)\int_{\mathbb R^3}\|\phi(x-u)\|^2du&=&\int_{\mathbb R^3}\int_{\mathbb R^3}\mathcal G^{A_1,A_2,A_3}_{\phi}\{f\}(w,u)K_{A^{-1}_3}^{\mu_4}(w_3,x_3)K_{A^{-1}_2}^{\mu_2}(w_2,x_2)\\
\label{a3}&&\qquad\qquad\qquad\times K_{A^{-1}_1}^{\mu_1}(w_1,x_1)\phi(x-u)dwdu.
\end{eqnarray}
Hence
\begin{eqnarray}
\nonumber f(x)\|\phi\|^2_{L^2(\mathbb R^3)}&=&\int_{\mathbb R^3}\int_{\mathbb R^3}\mathcal G^{A_1,A_2,A_3}_{\phi}\{f\}(w,u)K_{A^{-1}_3}^{\mu_4}(w_3,x_3)K_{A^{-1}_2}^{\mu_2}(w_2,x_2)\\
\label{a4}&&\qquad\qquad\qquad\times K_{A^{-1}_1}^{\mu_1}(w_1,x_1)\phi(x-u)dwdu.
\end{eqnarray}
Which completes the proof.
\end{proof}

 Prior to the next theorem we define the 3D-STLCT corresponding to the 3D-STFT\cite{1ostft} as
 \begin{definition}\label{def 3dstlct} Let $f,\phi\in L^2(\mathbb R^3)$ be a real-valued function, where $\phi$ is non-zero window function. Then 3D-STLCT is defined by
\begin{equation}\label{eqn 3dstlct}
\mathcal V^{A_1,A_2,A_3}_{\phi}\{f\}(w,u)=\dfrac{1}{2\pi\sqrt{2\pi|b_1b_2b_3|}}\int_{\mathbb R^3}f(x)\overline{\phi(x-u)}e^{\mu_1\xi_1}e^{\mu_1\xi_2}e^{\mu_1\xi_3}dx.
\end{equation}
\end{definition}

Now, we shall show that the STOLCT is related to 3D-STLCT.

\begin{theorem}\label{th 3dslct rel}Let $f,\phi\in L^2(\mathbb R^3)$ be a real-valued function, where $\phi$ is non-zero window function. Let $\mathcal V^{A_1,A_2,A_3}_{\phi}\{f\}(w,u)$ be the 3D-STLCT of  function $f$ with respect to $\phi.$ Then the following equation is satisfied
\begin{eqnarray}\label{eqn 3dstlct rel}
\nonumber \mathcal G^{A_1,A_2,A_3}_{\phi}\{f\}(w,u)&=&\dfrac{1}{4}\left\{(\mathcal V^{A_1,A_2,A_3}_{\phi}\{f\}(w,u)+\mathcal L_{A_1,A_2,A'_3}\{f\}(w,u))(1-\mu_3)\right.\\
\nonumber\qquad &&\left.+(\mathcal V^{A_1,A'_2,A_3}_{\phi}\{f\}(w,u)+\mathcal V^{A_1,A'_2,A'_3}_{\phi}\{f\}(w,u))(1+\mu_3)\right\}\\
\nonumber\qquad\quad&& +\dfrac{1}{4}\left\{(\mathcal V^{A_1,A_2,A_3}_{\phi}\{f\}(w,u)-\mathcal V^{A_1,A_2,A'_3}_{\phi}\{f\}(w,u))(1-\mu_3)\right.\\
\nonumber\qquad &&\left.+(\mathcal V^{A_1,A'_2,A_3}_{\phi}\{f\}(w,u)-\mathcal V^{A_1,A'_2,A'_3}_{\phi}\{f\}(w,u))(1+\mu_3)\right\}.\mu_5\\
\end{eqnarray}
where $A'_k=(a_k,-b_k,-c_k,d_k),\quad k=2,3.$
\end{theorem}
\begin{proof}
 From Definition \ref{def 3dstlct}, we have
 \begin{equation}\label{eqn b1}
\mathcal V^{A_1,A_2,A_3}_{\phi}\{f\}(w,u)=\dfrac{1}{2\pi\sqrt{2\pi|b_1b_2b_3|}}\int_{\mathbb R^3}f(x)\overline{\phi(x-u)}e^{\mu_1\xi_1}e^{\mu_1\xi_2}e^{\mu_1\xi_3}dx.
\end{equation}
Now for $A'_2=(a_2,-b_2,-c_2,d_2),$ then
\begin{equation}\label{eqn b2}
\mathcal V^{A_1,A'_2,A_3}_{\phi}\{f\}(w,u)=\dfrac{1}{2\pi\sqrt{2\pi|b_1b_2b_3|}}\int_{\mathbb R^3}f(x)\overline{\phi(x-u)}e^{\mu_1\xi_1}e^{-\mu_1\xi_2}e^{\mu_1\xi_3}dx.
\end{equation}
By equivalent definition of sine and cosine functions, we obtain
\begin{eqnarray}\label{b3}
\nonumber&&\frac{1}{2}\left(\mathcal V^{A_1,A_2,A_3}_{\phi}\{f\}(w,u)+\mathcal V^{A_1,A'_2,A_3}_{\phi}\{f\}(w,u)\right)\\ \
\nonumber&&\qquad\qquad=\dfrac{1}{2\pi\sqrt{2\pi|b_1b_2b_3|}}\int_{\mathbb R^3}f(x)\overline{\phi(x-u)}e^{\mu_1\xi_1}\cos{\xi_2}e^{\mu_1\xi_3}dx.\\ \
\end{eqnarray}
And
\begin{eqnarray}\label{b4}
\nonumber&&\frac{1}{2}\left(\mathcal V^{A_1,A'_2,A_3}_{\phi}\{f\}(w,u)-\mathcal V^{A_1,A_2,A_3}_{\phi}\{f\}(w,u)\right)\\ \
\nonumber&&\qquad\qquad=\dfrac{1}{2\pi\sqrt{2\pi|b_1b_2b_3|}}\int_{\mathbb R^3}f(x)\overline{\phi(x-u)}e^{\mu_1\xi_1}(-\mu_1\sin{\xi_2})e^{\mu_1\xi_3}dx.\\ \
\end{eqnarray}
On multiplying (\ref{b4}) from right by $\mu_3$ and using multiplication rules from Table \ref{table}, we have
\begin{eqnarray}\label{b5}
\nonumber&&\frac{1}{2}\left(\mathcal V^{A_1,A'_2,A_3}_{\phi}\{f\}(w,u)-\mathcal V^{A_1,A_2,A_3}_{\phi}\{f\}(w,u)\right)\mu_3\\ \
\nonumber&&\qquad\qquad=\dfrac{1}{2\pi\sqrt{2\pi|b_1b_2b_3|}}\int_{\mathbb R^3}f(x)\overline{\phi(x-u)}e^{\mu_1\xi_1}(\mu_2\sin{\xi_2})e^{-\mu_1\xi_3}dx.\\ \
\end{eqnarray}
Adding (\ref{b3}) and (\ref{b5}), we get
\begin{eqnarray}\label{b6}
\nonumber&&\frac{1}{2}\left(\mathcal V^{A_1,A_2,A_3}_{\phi}\{f\}(w,u)+\mathcal V^{A_1,A'_2,A_3}_{\phi}\{f\}(w,u)\right)\\ \
\nonumber&&+\frac{1}{2}\left(\mathcal V^{A_1,A'_2,A_3}_{\phi}\{f\}(w,u)-\mathcal V^{A_1,A_2,A_3}_{\phi}\{f\}(w,u)\right)\mu_3\\ \
\nonumber&&\qquad\qquad=\dfrac{1}{2\pi\sqrt{2\pi|b_1b_2b_3|}}\int_{\mathbb R^3}f(x)\overline{\phi(x-u)}e^{\mu_1\xi_1}e^{\mu_2\xi_2}e^{-\mu_1\xi_3}dx.\\ \
\end{eqnarray}
To simplify we introduce the following notation:
\begin{eqnarray}\label{b7}
\nonumber\mathbb V^{A_1,A_2,A_3}_{A_1,A'_2,A_3}(w,u)&=&\frac{1}{2}\left(\mathcal V^{A_1,A_2,A_3}_{\phi}\{f\}(w,u)+\mathcal V^{A_1,A'_2,A_3}_{\phi}\{f\}(w,u)\right)\\ \
\nonumber&&+\frac{1}{2}\left(\mathcal V^{A_1,A'_2,A_3}_{\phi}\{f\}(w,u)-\mathcal V^{A_1,A_2,A_3}_{\phi}\{f\}(w,u)\right)\mu_3.\\ \
\end{eqnarray}
Now for $A'_3=(a_3,-b_3,-c_3,d_3),$ then
\begin{eqnarray}
\nonumber\mathbb V^{A_1,A_2,A'_3}_{A_1,A'_2,A'_3}(w,u)&=&\frac{1}{2}\left(\mathcal V^{A_1,A_2,A'_3}_{\phi}\{f\}(w,u)+\mathcal V^{A_1,A'_2,A'_3}_{\phi}\{f\}(w,u)\right)\\
\label{b8}&&+\frac{1}{2}\left(\mathcal V^{A_1,A'_2,A'_3}_{\phi}\{f\}(w,u)-\mathcal V^{A_1,A_2,A'_3}_{\phi}\{f\}(w,u)\right)\mu_3.\\
\label{b9}&=&\dfrac{1}{2\pi\sqrt{2\pi|b_1b_2b_3|}}\int_{\mathbb R^3}f(x)\overline{\phi(x-u)}e^{\mu_1\xi_1}e^{\mu_2\xi_2}e^{\mu_3\xi_3}dx.
\end{eqnarray}
By following similar steps as before we get
\begin{eqnarray}\label{b10}
\nonumber&&\frac{1}{2}\left(\mathbb V^{A_1,A_2,A_3}_{A_1,A'_2,A_3}\{f\}(w,u)+\mathbb V^{A_1,A_2,A'_3}_{A_1,A'_2,A'_3}\{f\}(w,u)\right)\\ \
\nonumber&&\qquad\qquad=\dfrac{1}{2\pi\sqrt{2\pi|b_1b_2b_3|}}\int_{\mathbb R^3}f(x)\overline{\phi(x-u)}e^{\mu_1\xi_1}e^{\mu_2\xi_2}\cos{\xi_3}dx.\\ \
\end{eqnarray}
And
\begin{eqnarray}\label{b11}
\nonumber&&\frac{1}{2}\left(\mathbb V^{A_1,A_2,A_3}_{A_1,A'_2,A_3}\{f\}(w,u)-\mathbb V^{A_1,A_2,A'_3}_{A_1,A'_2,A'_3}\{f\}(w,u)\right)\\ \
\nonumber&&\qquad\qquad=\dfrac{1}{2\pi\sqrt{2\pi|b_1b_2b_3|}}\int_{\mathbb R^3}f(x)\overline{\phi(x-u)}e^{\mu_1\xi_1}e^{\mu_2\xi_2}(-\mu_1\sin{\xi_3})dx.\\ \
\end{eqnarray}
On multiplying (\ref{b11}) from right by $\mu_5$ and using multiplication rules from Table \ref{table}, we have
\begin{eqnarray}\label{b12}
\nonumber&&\frac{1}{2}\left(\mathbb V^{A_1,A_2,A_3}_{A_1,A'_2,A_3}\{f\}(w,u)-\mathbb V^{A_1,A_2,A'_3}_{A_1,A'_2,A'_3}\{f\}(w,u)\right)\mu_5\\ \
\nonumber&&\qquad\qquad=\dfrac{1}{2\pi\sqrt{2\pi|b_1b_2b_3|}}\int_{\mathbb R^3}f(x)\overline{\phi(x-u)}e^{\mu_1\xi_1}e^{\mu_2\xi_2}(\mu_4\sin{\xi_3})dx.\\ \
\end{eqnarray}
Adding (\ref{b10}) and (\ref{b12}), we get
\begin{eqnarray}\label{b13}
\nonumber&&\frac{1}{2}\left(\mathbb V^{A_1,A_2,A_3}_{A_1,A'_2,A_3}\{f\}(w,u)+\mathbb V^{A_1,A_2,A'_3}_{A_1,A'_2,A'_3}\{f\}(w,u)\right)\\ \
\nonumber&&+\frac{1}{2}\left(\mathbb V^{A_1,A_2,A_3}_{A_1,A'_2,A_3}\{f\}(w,u)-\mathbb V^{A_1,A_2,A'_3}_{A_1,A'_2,A'_3}\{f\}(w,u)\right)\mu_5\\ \
\nonumber&&\qquad\qquad=\dfrac{1}{2\pi\sqrt{2\pi|b_1b_2b_3|}}\int_{\mathbb R^3}f(x)\overline{\phi(x-u)}e^{\mu_1\xi_1}e^{\mu_2\xi_2}e^{\mu_4\xi_3}dx.\\ \
\end{eqnarray}
On substituting (\ref{b7}) and (\ref{b8}) in (\ref{b13}), we get the desired result.
\end{proof}

\section{\bf Uncertainty Inequalities for the STOLCT}
We know that in signal processing there are  different types of uncertainty principles in the QFT, QLCT and QOLCT domains. In \cite{li,ownoolct} authors investigate Heisenberg's uncertainty principle and Donoho-Stark's uncertainty principle, Pitt's inequality, logarithmic uncertainty inequality, Hausdorff-Young inequality and local uncertainty inequality for the octonion linear canonical transform and octonion offset linear canonical transform.
Recently  in \cite{stoft} authors establish Pitt's inequality, Lieb's inequality and logarithmic uncertainty principle for the STOFT.
Considering that the STOLCT is a generalized
version of the OLCT, it is natural and interesting to study  uncertainty principles of a real-valued function and its
STOLCT. So in this section we  shall investigate some uncertainty inequalities for the STOLCT.

\begin{lemma}\label{lem pq1} Let $f\in L^p(\mathbb{R}^3),\phi\in L^q\left( \mathbb{R}^3\right)$ and $ \frac{1}{p}+\frac{1}{q}=1,$  we have
\begin{equation}\label{eqn pq}
|\mathcal G_{\phi}^{A_1,A_2,A_3}\{f\}(w,u)|\leq \frac{1}{2\pi\sqrt{2\pi|b_{1}b_{2}b_3|}} \|f\|_{L^{p}(\mathbb R^3)}. \|\phi\|_{L^{q}(\mathbb R^3)}
\end{equation}
\end{lemma}
\begin{proof}
By the virtue of  H\"{o}lders inequality , we have
\begin{eqnarray*}
|\mathcal G_{\phi}^{A_1,A_2,A_3}\{f\}(w,u)|&=& \left|\int_{\mathbb R^3} f(t)\overline{\phi(x-u)}K^{\mu_1}_{A_1}(x_1,w_1)K^{\mu_2}_{A_2}(x_2,w_2)K^{\mu_4}_{A_3}(x_3,w_3)dx\right|\\
&\leq&\left(\int_{\mathbb R^3}\left| f(t)\overline{\phi(x-u)}K^{\mu_1}_{A_1}(x_1,w_1)K^{\mu_2}_{A_2}(x_2,w_2)K^{\mu_4}_{A_3}(x_3,w_3)\right|dx\right)\\
&=&\dfrac{1}{2\pi\sqrt{2\pi|b_{1}b_{2}b_3|}}\left(\int_{\mathbb R^3}\left|f(t)\overline{\phi(x-u)}\right|\right)dx\\
&\leq&\dfrac{1}{2\pi\sqrt{2\pi|b_{1}b_{2}b_3|}}\left(\int_{\mathbb{R}^3} \left| f(t)\right|^{p}{d}t\right)^{1/p}\left(\int_{\mathbb{R}^3} \left| \overline{\phi(t-u)}\right|^{q} {d}t\right)^{1/q}\\
&=&\frac{1}{2\pi\sqrt{2\pi|b_{1}b_{2}b_3|}} \|f\|_{L^{p}(\mathbb R^3)}. \|\phi\|_{L^{q}(\mathbb R^3)}.
\end{eqnarray*}
Which completes the proof.
\end{proof}
\begin{theorem}\label{local}
Let $\Omega$ be a measurable set $\subset\mathbb R^3\times\mathbb R^3$ and suppose that   $\phi,f\in L^2(\mathbb R^3)$ be two signals with $\|f\|_{L^{p}(\mathbb R^3)}=1=\|\phi\|_{L^{q}(\mathbb R^3)},$ with $\epsilon\ge0$ and
\begin{equation}\label{e}
\int\int_{\Omega}|\mathcal G_{\phi}^{A_1,A_2,A_3}\{f\}(w,u)|^2dwdu\ge 1-\epsilon.
\end{equation}
We have $2\pi\sqrt{2\pi|b_1b_1b_3|}(1-\epsilon)\le m(\Omega),$ where $m(\Omega)$ is Lebesgue measure of $\Omega.$
\end{theorem}

\begin{proof}
From Lemma \ref{local}, we have
\begin{equation}\label{eqn u}
\|\mathcal G_{\phi}^{A_1,A_2,A_3}\{f\}(w,u)\|_{L^{\infty}(\mathbb R^3)}\leq \frac{1}{2\pi\sqrt{2\pi|b_{1}b_{2}b_3|}} \|f\|_{L^{p}(\mathbb R^3)}. \|\phi\|_{L^{q}(\mathbb R^3)}.\\\
\end{equation}
On inserting (\ref{eqn u}) in (\ref{e}), we obtain
\begin{eqnarray*}
1-\epsilon&\le&\int\int_{\Omega}|\mathcal G_{\phi}^{A_1,A_2,A_3}\{f\}(w,u)|^2dwdu\\\
&\le&m(\Omega)\|\mathcal G_{\phi}^{A_1,A_2,A_3}\{f\}(w,u)\|_{L^{\infty}(\mathbb R^3)}\\\
&\le&m(\Omega)\frac{1}{2\pi\sqrt{2\pi|b_{1}b_{2}b_3|}} \|f\|_{L^{p}(\mathbb R^3)}. \|\phi\|_{L^{q}(\mathbb R^3)}\\\
&=&\frac{m(\Omega)}{2\pi\sqrt{2\pi|b_{1}b_{2}b_3|}}
\end{eqnarray*}
 implies
${2\pi\sqrt{2\pi|b_{1}b_{2}b_3|}}(1-\epsilon)\le m(\Omega).$
Which completes the proof.
\end{proof}

\begin{theorem}[Lieb's inequality fo the STOLCT] Let $2\leq p\le\infty$ and $\phi\in L^2(\mathbb R^3)$ be a non-zero real-valued window function. For every real-valued signal $f\in L^2(\mathbb R^3),$ we have
\begin{equation}\label{eqn leibs}
\|\mathcal G^{A_1,A_2,A_3}_{\phi}\{f\}(w,u)\|_{L^q(\mathbb R^3)}\leq\frac{|b_1b_2|^{\frac{-q}{2}+1}}{(2\pi)^{q+\frac{1}{2}}|b_3|^{\frac{1}{2}}}E_{p,q}\|f\|_{L^2(\mathbb R^3)}\|\phi\|_{L^2(\mathbb R^3)}
\end{equation}
where $E_{p,q}=\left(\frac{4}{q}\right)^{\frac{1}{q}}\left(\frac{4}{p}\right)^{\frac{1}{p}}$ and $\frac{1}{p}+\frac{1}{q}=1.$
\end{theorem}
{\bf Before proving this theorem we shall recall the following lemma.}

\begin{lemma}[Hausdorff-Young inequality for QLCT]\label{lem hdrfqlct}\cite{mb} For $1\le p\le 2$ and $\frac{1}{p}+\frac{1}{q}=1,$ we have
\begin{equation}\label{eqnhdrof qlct}
\|\mathcal L^{A_1,A_2}_{\mu_1,\mu_2,}\{f\}(w)\|_q\le  \frac{|b_1b_2|^{-\frac{1}{2}+\frac{1}{q}}}{2\pi}\|f(x)\|_p.
\end{equation}
\end{lemma}

\begin{lemma}[Hausdorff-Young inequality for OLCT]\label{lem hdrfolct} For $1\le p\le 2$ and $\frac{1}{p}+\frac{1}{q}=1,$ we have
\begin{equation}\label{eqnhdrof olct}
\|\mathcal L^{A_1,A_2,A_3}_{\mu_1,\mu_2,\mu_4}\{f\}(w)\|_q\le \frac{|b_1b_2|^{-\frac{1}{2}+\frac{1}{q}}}{(2\pi)^{\frac{1}{2q}+1}|b_3|^{\frac{1}{2q}}}\|f(x)\|_p .
\end{equation}
\end{lemma}
\begin{proof} We omit proof as it is similar to the proof of the Theorem 4.3\cite{ownoolct}.
\end{proof}

{\bf Proof of main theorem}
\begin{proof} We have $ \mathcal G^{A_1,A_2,A_3}_{\phi}\{f\}(w,u)=\mathcal L^{A_1,A_2,A_3}_{\mu_1,\mu_2,\mu_4}\{f(x)\overline{\phi(x-u)}\}(w),$ by Lemma \ref{lem hdrfolct}, we have
\begin{eqnarray*}
\nonumber \left(\int_{\mathbb R^3}\left|\mathcal G^{A_1,A_2,A_3}_{\phi}\{f\}(w,u)\right|^qdw\right)^{\frac{1}{q}}&=&\left(\int_{\mathbb R^3}\left|\mathcal L^{A_1,A_2,A_3}_{\mu_1,\mu_2,\mu_4}\{f(x)\overline{\phi(x-u)}\}(w)\right|^qdw\right)^{\frac{1}{q}}\\\
\nonumber&&\le\frac{|b_1b_2|^{-\frac{1}{2}+\frac{1}{q}}}{(2\pi)^{\frac{1}{2q}+1}|b_3|^{\frac{1}{2q}}}\|f(x)\overline{\phi(x-u)}\|_p \\\
\nonumber&=&\frac{|b_1b_2|^{-\frac{1}{2}+\frac{1}{q}}}{(2\pi)^{\frac{1}{2q}+1}|b_3|^{\frac{1}{2q}}}\left(\int_{\mathbb R^3}|f(x)\overline{\phi(x-u)}|^pdx\right)^{\frac{1}{p}}\\\
\end{eqnarray*}
then,
\begin{eqnarray*}\label{c1}
\nonumber\int_{\mathbb R^3}\left|\mathcal G^{A_1,A_2,A_3}_{\phi}\{f\}(w,u)\right|^qdw&&\\
\nonumber\qquad\qquad&&\le\frac{|b_1b_2|^{\frac{-q}{2}+1}}{(2\pi)^{q+\frac{1}{2}}|b_3|^{\frac{1}{2}}}\left(|f|^p\ast|\tilde\phi|^p(u)\right)^{\frac{q}{p}},\\\
\end{eqnarray*}
where $\tilde\phi(x)=\phi(-x)$
Thus
\begin{eqnarray*}\label{c2}
\|\mathcal G^{A_1,A_2,A_3}_{\phi}\{f\}(w,u)\|_q&=&\left(\int_{\mathbb R^3}\left(\int_{\mathbb R^3}\left|\mathcal G^{A_1,A_2,A_3}_{\phi}\{f\}(w,u)\right|^qdw\right)du\right)^{\frac{1}{q}}\\\
\nonumber&&\le\frac{|b_1b_2|^{\frac{-q}{2}+1}}{(2\pi)^{q+\frac{1}{2}}|b_3|^{\frac{1}{2}}}\left(\int_{\mathbb R^3}\left(|f|^p\ast|\tilde\phi|^p(u)\right)^{\frac{q}{p}}du\right)^{\frac{1}{q}}\\\
\nonumber&=&\frac{|b_1b_2|^{\frac{-q}{2}+1}}{(2\pi)^{q+\frac{1}{2}}|b_3|^{\frac{1}{2}}}\left\||f|^p\ast|\tilde\phi|^p\right\|^{\frac{1}{p}}_{L^{\frac{q}{p}}(\mathbb R^3)}\\\
\end{eqnarray*}
If $k=\frac{2}{p}$, $l=\frac{q}{p}$ and $\frac{1}{k}+\frac{1}{k'}=1,$ $\frac{1}{l}+\frac{1}{l'}=1$ then $\frac{1}{k}+\frac{1}{k}=1+\frac{1}{l}$ and $|f|^p$ ,$|\tilde\phi|^p\in L^2(\mathbb R^3)$ and by Young inequality, we have
\begin{equation*}\label{c3}
\left\||f|^p\ast|\tilde\phi|^p\right\|_{L^l(\mathbb R^3)}\le B^4_kB_l^2\||f|^p\|_{L^k(\mathbb R^3)}\||\tilde\phi|^p\|_{L^l(\mathbb R^3)},
\end{equation*}
where $B_s=\left(\frac{s^{\frac{1}{s}}}{s'^{\frac{1}{s'}}}\right)^{\frac{1}{2}}$ , $\frac{1}{s}+\frac{1}{s'}=1.$
However
\begin{equation*}
\||f|^p\|_{L^k(\mathbb R^3)}=\left(\int_{\mathbb R^3}|f(x)|^{p.\frac{2}{p}}dx\right)^{\frac{p}{2}}=\|f\|^p_{L^2(\mathbb R^3)},
\end{equation*}
and
\begin{equation*}
\||\tilde\phi|^p\|_{L^k(\mathbb R^3)}=\left(\int_{\mathbb R^3}|\phi(x-u)|^{p.\frac{2}{p}}dx\right)^{\frac{p}{2}}=\|\phi\|^p_{L^2(\mathbb R^3)}.
\end{equation*}
Hence
\begin{eqnarray*}
\nonumber\|\mathcal G^{A_1,A_2,A_3}_{\phi}\{f\}(w,u)\|_q&\le&\frac{|b_1b_2|^{\frac{-q}{2}+1}}{(2\pi)^{q+\frac{1}{2}}|b_3|^{\frac{1}{2}}}\left\||f|^p\ast|\tilde\phi|^p\right\|^{\frac{1}{p}}_{L^{\frac{q}{p}}(\mathbb R^3)}\\\
&\le&\frac{|b_1b_2|^{\frac{-q}{2}+1}}{(2\pi)^{q+\frac{1}{2}}|b_3|^{\frac{1}{2}}}\left(B^4_kB_l'^2\|f\|^p_{L^2(\mathbb R^3)}\|\phi\|^p_{L^2(\mathbb R^3)}\right)^{\frac{1}{p}}\\\
&=&\frac{|b_1b_2|^{\frac{-q}{2}+1}}{(2\pi)^{q+\frac{1}{2}}|b_3|^{\frac{1}{2}}}B^{\frac{4}{p}}_kB_l'^{\frac{2}{p}}\|f\|_{L^2(\mathbb R^3)}\|\phi\|_{L^2(\mathbb R^3)}\\\
&=&\frac{|b_1b_2|^{\frac{-q}{2}+1}}{(2\pi)^{q+\frac{1}{2}}|b_3|^{\frac{1}{2}}}E_{p,q}\|f\|_{L^2(\mathbb R^3)}\|\phi\|_{L^2(\mathbb R^3)}.\\\
\end{eqnarray*}
Which completes the proof.
\end{proof}

\begin{lemma}[Logarithmic uncertainty principle for the OLCT]\label{lem logolct} Let $f\in \mathcal S(\mathbb R^3),$ then the following inequality is satisfied:
\begin{equation}\label{eqn logolct}
2\pi|b_3|\int_{{\mathbb R}^3}\ln\left|w\right|\left|{\mathcal L}^{A_1,A_2,A_3}_{\mu_1,\mu_2,\mu_3}\{f\}( w)\right|^2d w + \int_{{\mathbb R}^3}\ln|x| |f(x)|^2dx\ge E \int_{{\mathbb R}^3} |f(x)|^2dx\\ \
\end{equation}
with $D=\ln(2)+\Gamma'(\frac{1}{2})/\Gamma(\frac{1}{2}).$
\end{lemma}
\begin{proof}
We avoid proof as it follows by using the procedure of the Theorem 4.2\cite{ownoolct}. Also see Theorem 5\cite{w2}.
\end{proof}

\begin{theorem}[Logarithmic uncertainty principle for the STOLCT]\label{thm logstolct} Let $f,\phi\in \mathcal S(\mathbb R^3),$ then the following inequality is satisfied:
\begin{eqnarray}\label{eqn logstolct}
\nonumber 2\pi|b_3|\int_{{\mathbb R}^3}\int_{{\mathbb R}^3}\ln\left|w\right|\left|{\mathcal G}^{A_1,A_2,A_3}_{\phi}\{f\}( w,u)\right|^2dwdu + \|\phi\|^2_{L^2({\mathbb R}^3)}\int_{{\mathbb R}^3}\ln|x| |f(x)|^2dx&&\\ \
\nonumber\ge E  \|f\|^2_{L^2({\mathbb R}^3)}\|\phi\|^2_{L^2({\mathbb R}^3)}&&\\ \
\end{eqnarray}
with $D=\ln(2)+\Gamma'(\frac{1}{2})/\Gamma(\frac{1}{2}).$
\end{theorem}
\begin{proof}
As $f,\phi\in \mathcal S(\mathbb R^3)$ implies $f(x)\overline{\phi(x-u)}=h(x,u)\in  \mathcal S(\mathbb R^3).$ Thus  replacing $f(x)$ by $h(x,u)$ in Lemma \ref{lem logolct}, we obtain
\begin{equation}\label{l1}
2\pi|b_3|\int_{{\mathbb R}^3}\ln\left|w\right|\left|{\mathcal L}^{A_1,A_2,A_3}_{\mu_1,\mu_2,\mu_3}\{h\}( w)\right|^2d w + \int_{{\mathbb R}^3}\ln|x| |h(x,u)|^2dx\ge E \int_{{\mathbb R}^3} |h(x,u)|^2dx.\\ \
\end{equation}
On integrating both sides of (\ref{l1}) with respect to du, we have
\begin{eqnarray}\label{l2}
\nonumber 2\pi|b_3|\int_{{\mathbb R}^3}\int_{{\mathbb R}^3}\ln\left|w\right|\left|{\mathcal L}^{A_1,A_2,A_3}_{\mu_1,\mu_2,\mu_3}\{h\}( w)\right|^2d wdu + \int_{{\mathbb R}^3}\int_{{\mathbb R}^3}\ln|x| |h(x,u)|^2dxdu&&\\ \
\nonumber\ge E \int_{{\mathbb R}^3}\int_{{\mathbb R}^3} |h(x,u)|^2dxdu.&&\\ \
\end{eqnarray}
Now using (\ref{eqn stolct new}) in (\ref{l2}), we get
\begin{eqnarray}\label{l3}
\nonumber 2\pi|b_3|\int_{{\mathbb R}^3}\int_{{\mathbb R}^3}\ln\left|w\right|\left|{\mathcal G}^{A_1,A_2,A_3}_{\phi}\{f\}( w,u)\right|^2dwdu + \int_{{\mathbb R}^3}\int_{{\mathbb R}^3}\ln|x| |f(x)|^2|\phi(x-u)|^2dxdu&&\\ \
\nonumber\ge E \int_{{\mathbb R}^3}\int_{{\mathbb R}^3} |f(x)|^2|\phi(x-u)|^2dxdu.&&\\ \
\end{eqnarray}
 Which implies
 \begin{eqnarray}\label{l4}
\nonumber 2\pi|b_3|\int_{{\mathbb R}^3}\int_{{\mathbb R}^3}\ln\left|w\right|\left|{\mathcal G}^{A_1,A_2,A_3}_{\phi}\{f\}( w,u)\right|^2dwdu + \int_{{\mathbb R}^3}|\phi(x-u)|^2du\int_{{\mathbb R}^3}\ln|x| |f(x)|^2dx&&\\ \
\nonumber\ge E \int_{{\mathbb R}^3}|f(x)|^2dx\int_{{\mathbb R}^3} |\phi(x-u)|^2du.&&\\ \
\end{eqnarray}
  Hence
  \begin{eqnarray}\label{l5}
\nonumber 2\pi|b_3|\int_{{\mathbb R}^3}\int_{{\mathbb R}^3}\ln\left|w\right|\left|{\mathcal G}^{A_1,A_2,A_3}_{\phi}\{f\}( w,u)\right|^2dwdu + \|\phi\|^2_{L^2({\mathbb R}^3)}\int_{{\mathbb R}^3}\ln|x| |f(x)|^2dx&&\\ \
\nonumber\ge E  \|f\|^2_{L^2({\mathbb R}^3)}\|\phi\|^2_{L^2({\mathbb R}^3)}.&&\\ \
\end{eqnarray}
Which completes the proof.
\end{proof}
\subsection{\bf Convolution Theorem for the STOLCT}
The convolution  has wide has wide range of  applications in various
areas of Mathematics like linear algebra, numerical analysis and signal processing. So in this subsection we establish the convolution theorem for the STOLCT on the bases of the classical convolution operator (\ref{con op}).

\begin{theorem}[Convolution]\label{th conv}
Let $\phi,\psi\in L^2(\mathbb R^3)$ be two non-zero real-valued window functions then for any real-valued functions $f,g\in L^2(\mathbb R^3)$, we have
\begin{eqnarray}\label{eqn con}
&&\nonumber\mathcal G^{A_1,A_2,A_3}_{\phi\ast\psi}\{f\ast g\}(w,u)\\\
\nonumber&&=\int_{\mathbb R^3}\mathcal G^{A_1,A_2,A_3}_{\psi}\{g\}(w,m)G^\phi_{eee}(w,u-m)+ G^{A_1,A_2,A_3}_{\psi}\{g\}(t,m)G^\phi_{oee}(w,u-m)\mu_1\\\
\nonumber&&\qquad+G^{A_1,A_2,A_3}_{\psi}\{g\}(s,m)G^\phi_{eoe}(w,u-m)\mu_2+G^{A_1,A_2,A_3}_{\psi}\{g\}(t,m)G^\phi_{ooe}(w,u-m)\mu_3\\\
\nonumber&&\qquad+G^{A_1,A_2,A_3}_{\psi}\{g\}(w,m)G^\phi_{eeo}(w,u-m)\mu_4+G^{A_1,A_2,A_3}_{\psi}\{g\}(s,m)G^\phi_{oeo}(w,u-m)\mu_5\\\
\nonumber&&\qquad+G^{A_1,A_2,A_3}_{\psi}\{g\}(t',m)G^\phi_{eoo}(w,u-m)\mu_6+G^{A_1,A_2,A_3}_{\psi}\{g\}(t',m)G^\phi_{ooo}(w,u-m)\mu_7dm,\\\
\end{eqnarray}
where $t=(w_1,-w_2,-w_3)\in\mathbb R^3$, $s=(w_1,w_2,-w_3)\in\mathbb R^3,$ $t'=(-w_1,w_2,-w_3)\in\mathbb R^3$ and $G^\phi_{lmn},l,m,n\in\{e,o\}$ are   given by equations  (\ref{e1}) to (\ref{e8}).
\end{theorem}

\begin{proof}
By Definition \ref{def stolct}\ we have by classic convolution operator
\begin{eqnarray}\label{aaa}
&&\nonumber\mathcal G^{A_1,A_2,A_3}_{\phi\ast\psi}\{f\ast g\}(w,u)\\\
\nonumber&&=\int_{\mathbb R^3}\left(\int_{\mathbb R^3}f(y)g(x-y)dy\right).\int_{\mathbb R^3}\left(\int_{\mathbb R^3}\overline{\phi(z)\psi(x-u-z)}dz\right)\\\
\nonumber&&\times K^{\mu_1}_{A_1}(x_1,w_1)K^{\mu_2}_{A_2}(x_2,w_2)K^{\mu_4}_{A_3}(x_3,w_3)\\\
\end{eqnarray}
\end{proof}
Setting $(q_1,q_2,q_3)=q=x-y,$ $(m_1,m_2,m_3)=m=u+z-y,$ in (\ref{aaa}), we obtain
\begin{eqnarray}\label{aaa1}
&&\nonumber\mathcal G^{A_1,A_2,A_3}_{\phi\ast\psi}\{f\ast g\}(w,u)\\\
\nonumber&&=\int_{\mathbb R^9}f(y)g(q)\overline{\phi(y-(u-m))}.\overline{\psi(q-m)}K^{\mu_1}_{A_1}(q_1+y_1,w_1)\\\
\nonumber&&\qquad\times K^{\mu_2}_{A_2}(q_2+y_2,w_2)K^{\mu_4}_{A_3}(q_3+y_3,w_3)dqdydm\\\
\nonumber&&=\int_{\mathbb R^6}f(y)\overline{\phi(y-(u-m))}\left(\int_{\mathbb R^3}g(q)\overline{\psi(q-m)}K^{\mu_1}_{A_1}(q_1+y_1,w_1)\right.\\\
\nonumber&&\qquad\times \left. K^{\mu_2}_{A_2}(q_2+y_2,w_2)K^{\mu_4}_{A_3}(q_3+y_3,w_3)dq\right)dydm\\\
\end{eqnarray}
Now,

 let $$\xi_k={\frac{{1}}{2b_k}\big[a_kq^2_k-2q_kw_k+d_kw^2_k-\frac{\pi}{2}\big]}$$ and $$\gamma_k={\frac{{1}}{2b_k}\big[a_ky^2_k+2a_kq_ky_k-2y_1w_1\big]},\quad k=1,2,3$$
 then
\begin{eqnarray}\label{kt}
\nonumber K^{\mu_1}_{A_1}(q_1+y_1,w_1)K^{\mu_2}_{A_2}(q_2+y_2,w_2)K^{\mu_4}_{A_3}(q_3+y_3,w_3)&&\\\
\nonumber&&=\frac{1}{2\pi\sqrt{2\pi|b_1b_2b_3|}}e^{\mu_1(\xi_1+\gamma_1)}e^{\mu_2(\xi_2+\gamma_2)}e^{\mu_4(\xi_3+\gamma_3)}\\\
\end{eqnarray}
Now by applying multiplication rules of Table \ref{table}, we obtain

\begin{eqnarray}\label{kt2}
\nonumber e^{\mu_1(\xi_1+\gamma_1)}e^{\mu_2(\xi_2+\gamma_2)}e^{\mu_4(\xi_3+\gamma_3)}&&\\\
\nonumber&&=((e^{\mu_1\xi_1}.\cos(\gamma_1)).(e^{\mu_2\xi_2}.\cos(\gamma_2))).(e^{\mu_4\xi_3}.\cos(\gamma_3))\\\
\nonumber&&\qquad+((e^{\mu_1\xi_1}.\mu_1\sin(\gamma_1)).(e^{\mu_2\xi_2}.\cos(\gamma_2))).(e^{\mu_4\xi_3}.\cos(\gamma_3))\\\
\nonumber&&\qquad+((e^{\mu_1\xi_1}.\cos(\gamma_1)).(e^{\mu_2\xi_2}.\mu_2\sin(\gamma_2))).(e^{\mu_4\xi_3}.\cos(\gamma_3))\\\
\nonumber&&\qquad+((e^{\mu_1\xi_1}.\mu_1\sin(\gamma_1)).(e^{\mu_2\xi_2}.\mu_2\sin(\gamma_2))).(e^{\mu_4\xi_3}.(\cos(\gamma_3))\\\
\nonumber&&\qquad+((e^{\mu_1\xi_1}.\cos(\gamma_1)).(e^{\mu_2\xi_2}.\cos(\gamma_2))).(e^{\mu_4\xi_3}.\mu_4\sin(\gamma_3))\\\
\nonumber&&\qquad+((e^{\mu_1\xi_1}.\mu_1\sin(\gamma_1)).(e^{\mu_2\xi_2}.\cos(\gamma_2))).(e^{\mu_4\xi_3}.\mu_4\sin(\gamma_3))\\\
\nonumber&&\qquad+((e^{\mu_1\xi_1}.\cos(\gamma_1)).(e^{\mu_2\xi_2}.\mu_2\sin(\gamma_2))).(e^{\mu_4\xi_3}.\mu_4\sin(\gamma_3))\\\
\nonumber&&\qquad+((e^{\mu_1\xi_1}.\mu_1\sin(\gamma_1)).(e^{\mu_2\xi_2}.\mu_2\sin(\gamma_2))).(e^{\mu_4\xi_3}.\mu_4\sin(\gamma_3))\\\
\end{eqnarray}
On substituting (\ref{kt2}),(\ref{kt}) in (\ref{aaa1}) and noting $t=(w_1,-w_2,-w_3)\in\mathbb R^3$, $s=(w_1,w_2,-w_3)\in\mathbb R^3,$ $t'=(-w_1,w_2,-w_3)\in\mathbb R^3$, we get the desired result.

\section{\bf Potential Applications}
As for as generalization of transformation in to octonion algebra, the
OLCT transforms a octonion 3D signal into a octonion-valued frequency domain
signal, which is an effective processing tool for signal, image and  color analysis. The hypercomplex LCT that treats the mutichannel signls as a algebraic whole without losing the spectral relations for color image processing. But there is drawback that hypercomplex LCT canot reveal the local information of a signal due to its global kernel. The STOLCT is a new tool for time frequncy analysis which overcomes this drawback by using a sliding window. Another potential application of STOLCT  is that it can also reconstruct each monocomponent mode
from a multicomponent signal.\\
 Moreover, the uncertainty principle makes a tradeoff
between temporal and spectral resolutions unavoidable, i.e. the new uncertainty
principles for the STOLCT describe the relation of one octonion-valued
signal in spatial and another octonion-valued signal in frequency domain. They could
further contribute to solving problems of signal processing, optics, color image processing, quantum mechanics,  electrodynamics,
electromagnetism, etc.\\
 In addition, Lieb’s uncertainty principle
for the STOLCT could analyze the non-stationary signal and time-varying
system, which has a significant application in the study of signal local frequency spectrum.

 \section{Conclusions}\ \\

In this paper, we introduce octonion linear canonical transform of real-valued functions. Further more keeping in mind the varying frequencies, we used the proposed transform to generate a new transform called short-time octonion linear canonical transform (STOLCT). Then, the various properties of the proposed
STOLCT are explored, such as linearity, inversion formulas,
decomposition into components of different parity and relation with the 3D-STLCT. Furthermore,
Lieb’s inequality, logarithmic uncertainty inequality   associated with the STOLCT are investigated. Also based on classical convolution operation,  the convolution theorem for
the STOLCT is derived. Finally, some potential applications of the STLCT are presented.In
our future works, we will discuss the physical significance and engineering background of this paper.


\begin{thebibliography}{99}
\bibitem{stoft} Gao,W.B and Li, B.Z: Octonion short-time Fourier transform for
time-frequency representation and its applications. IEEE Transactions on signal processing.
\end{thebibliography}
\end{document}